\pdfoutput=1
\documentclass{amsart}

\usepackage{amsmath,amssymb}
\usepackage{microtype}
\usepackage{extarrows}
\usepackage{tikz}
\usetikzlibrary{decorations.markings}
\usepackage {mathpartir}
\usepackage{stmaryrd}
\usepackage{amsfonts}
\usepackage [all,2cell,pdf,v2]{xy}
\UseAllTwocells
\usepackage{bbold}
\usepackage{hyperref}

\newtheorem{thm}{Theorem}[section]
\newtheorem{prop}[thm]{Proposition}
\newtheorem{example}[thm]{Example}
\theoremstyle{definition}
\newtheorem{lemma}[thm]{Lemma}
\newtheorem{prob}[thm]{Problem}
\newtheorem{constrInternal}[thm]{Construction}
\newenvironment{construction}[1]
      {\begin{constrInternal}[{for Problem~\ref{#1}}]}
      {\qed\end{constrInternal}}

\newcommand{\op}{\mathsf{op}}
\newcommand{\ops}{\mathit{op}}
\newcommand{\core}{\mathsf{core}}
\newcommand{\cores}{\mathit{core}}
\newcommand{\homs}{\mathit{hom}}

\newcommand{\Id}{\mathsf{Id}}
\newcommand{\homt}{\mathsf{hom}}

\newcommand{\C}{\mathcal{C}}
\newcommand{\T}{\mathcal{T}}
\newcommand{\D}{\mathcal{D}}
\newcommand{\E}{\mathcal{E}}
\newcommand{\ob}{\text{ob}}
\newcommand{\dom}{\normalfont\textsc{dom}}
\newcommand{\cod}{\textsc{cod}}

\newcommand{\two}{\mathbb{2}}
\newcommand{\Cat}{{\mathcal{C}at}}
\newcommand{\Set}{{\mathcal{S}et}}
\newcommand{\CAT}{\mathcal{CAT}}
\newcommand{\type}{\textsf{\scriptsize{TYPE}}}
\newcommand{\sci}{\mathsf{i}}

\newcommand{\scer}{\mathsf{e_R}}
\newcommand{\scel}{\mathsf{e_L}}

\title{Towards a directed homotopy type theory}

\date{\today}

\author{Paige Randall North}

\begin{document}

\maketitle

\begin{abstract}
In this paper, we present a directed homotopy type theory for reasoning synthetically about (higher) categories, directed homotopy theory, and its applications to concurrency. 

We specify a new `homomorphism' type former for Martin-L\"of type theory which is roughly analogous to the identity type former originally introduced by Martin-L\"of. The homomorphism type former is meant to capture the notions of morphism (from the theory of categories) and directed path (from directed homotopy theory) just as the identity type former is known to capture the notions of isomorphism (from the theory of groupoids) and path (from homotopy theory). 

Our main result is an interpretation of these homomorphism types into $\Cat$, the category of small categories. There, the interpretation of each homomorphism type $\homt_\C(a,b)$ is indeed the set of morphisms between the objects $a$ and $b$ of the category $\C$.

We end the paper with an analysis of the interpretation in $\Cat$ with which we argue that our homomorphism types are indeed the directed version of Martin-L\"of's identity types. 
 \end{abstract}
\tableofcontents
\section{Introduction}

Martin-L\"of type theory, together with its identity type, is now often described as a synthetic theory of higher groupoids. This rich structure of the identity type was first discovered by Hofmann and Streicher in their disproof of the \emph{uniqueness of identity proofs} \cite{HS96}. There, the authors construct an interpretation of Martin-L\"of's identity type into the category of groupoids where for any terms $a,b$ in a type $\mathcal G$, the identity type $\Id_{\mathcal G}(a,b)$ is interpreted to be the set of morphisms $a \to b$ in the groupoid interpreting $\mathcal G$.

The full extent of this structure was later made explicit in Voevodsky's interpretation of Martin-L\"of type theory with the identity type (amongst other types) into the category of Kan complexes, the objects of which represent spaces or $\infty$-groupoids \cite{kl12}. In this interpretation, the identity type $\Id_K(a,b)$ of any two points $a,b$ in a Kan complex $K$ is itself a Kan complex, or space.

Following these concrete interpretations, others (e.g. \cite{War08, GG08, Nor}) have compared Martin-L\"of's specification of the identity types with the part of Quillen's axiomatization of homotopy theory that is today called the theory of weak factorization systems \cite{Qui67, Bro73}. In \cite{Nor}, the author has shown that in any finitely complete category, models of Martin-L\"of's identity type are in correspondence with weak factorization systems with two stability properties.

Such results describe a fascinating perspective on the theory of higher groupoids. However, one might argue that this theory is very well understood already by topologists since it has concrete models which have been studied for the better part of a century.
However, if one thinks of the theory of higher groupoids as being a theory of (reversible) paths, then one might wonder whether there is also a type theory which describes a theory of irreversible, or \emph{directed}, paths. One might hope that such a type theory would shed light on the less well understood theories of higher categories and directed spaces. Both of these theories have a notion of directed paths at their core, and for both of these theories, there are many competing concrete models. Thus, having a synthetic theory of directed paths might shed light on this situation. 

In this paper, we present the core of such a type theory. Our long-term goal is to use this type theory to prove and verify theorems about higher categories, directed homotopy theory, and concurrency.

\subsection{Previous work}

Several other type theories with semantics in category theories have been proposed. In \cite{HL}, the authors construct a type theory that gives a synthetic theory of $1$-categories. In \cite{RS}, the authors construct a type theory that gives a theory of $(\infty, 1)$-categories. Other works-in-progress \cite{Nuyts, War} do hope to give a theory of $\infty$-categories. 

In this paper, we aim to give a syntax which not only describes $\infty$-categories but also directed spaces. We have three design criteria: (1) that our notion of directed paths is introduced via a type former and not as a judgement, so that it can be iterated as Martin-L\"of's identity type can be, (2) that our theory allows us to transport terms of dependent types along directed paths, and (3) that our theory does \emph{not} always allow us to transport terms of the type of directed paths along identities, since such a transportation does not hold in categories of directed spaces. 

The syntaxes of \cite{HL} and \cite{Nuyts} do not meet our first criterion since their directed paths are introduced as a judgement and not as a type former. This seems to make the theory somewhat unwieldy: in \cite{HL}, only the theory of $1$-categories is introduced and \cite{Nuyts} includes no semantics. The theory of \cite{War} does not meet our second criterion, and the theory of \cite{RS} does not meet our third criterion.

\subsection{Organization}

In Section \ref{sec:syntax}, we present the syntax for our type theory, which consists of one new type former $\homt$ within Martin-L\"of type theory. In Section \ref{sec:semantics}, we interpret the rules of this new type former $\homt$ within a comprehension category derived from the category $\Cat$ of small categories. Even though the objects of $\Cat$ are simpler than those of our intended semantics (higher categories and directed spaces), we see this as an important test case since its theory is much better understood than those of of our intended semantics.
In Section \ref{sec:dhott}, we briefly describe directed homotopy theory and its connection with concurrency, and we sketch an interpretation of our type former $\homt$ in a category of directed spaces.

\section{The syntax}
\label{sec:syntax}

In this section, we define a type of homomorphisms within standard Martin-L\"of type theory \cite{Mar82}. We only work in the bare type theory. That is, we assume that we have all the `general rules' of \cite{Mar82}, but none of the type-forming rules given there.

First, we require that for every type $T$, there are two new types
\[
\inferrule
  { \Gamma \vdash T \ \type}
  { \Gamma \vdash T^{\core} \ \type}
   \hspace{3em}
\inferrule
{ \Gamma \vdash T \ \type}
{ \Gamma \vdash T^{\op} \ \type}
\]
 together with functions.
\[
\inferrule
  { \Gamma \vdash T \ \type \\ \Gamma \vdash t: T^{\core}}
  { \Gamma \vdash \sci (t): T}
   \hspace{3em}
\inferrule
{ \Gamma \vdash T \ \type \\ \Gamma \vdash t: T^{\core}}
{ \Gamma \vdash \sci^\op (t): T^\op}
\]

Now we can define our type of homomorphisms.

\begin{align*}
&\homt \textsc{ formation}: & & \homt \textsc{ introduction}:  \\
&\quad \inferrule
{ \Gamma \vdash T \ \type \\  \Gamma \vdash s : T^\op \\ \Gamma \vdash t : T}
{\Gamma \vdash \homt_T(s,t) \ \type} 
& & \quad \inferrule
{ \Gamma \vdash T \ \type \\  \Gamma \vdash t : T^\core }
{\Gamma \vdash 1_t : \homt_T(\sci^\op t,\sci t) \ \type} 
\end{align*}
\begin{align*}
& \textsc{right } \homt \textsc{ elimination}: \\
& \quad \inferrule
{ \Gamma \vdash T \ \type \\  
\Gamma, s: T^\core \vdash \Theta(s) \ \type \\ 
\Gamma, s: T^\core, t: T , f: \homt_T(\sci^\op s, t), \theta: \Theta(s) \vdash D(f,\theta) \ \type \\ 
\Gamma, s:  T^\core, \theta : \Theta(s) \vdash d(\theta) : D(1_s, \theta) 
}
{\Gamma, s: T^\core, t: T , f: \homt_T(\sci^\op s, t), \theta:\Theta(s) \vdash \scer(d,f,\theta) : D(f,\theta)}
\end{align*}
\begin{align*}
& \textsc{left }\homt \textsc{ elimination}: \\
& \quad
 \inferrule
{ \Gamma \vdash T \ \type \\  
\Gamma, s: T^\core  \vdash \Theta(s) \ \type\\
\Gamma, s: T^\op, t: T^\core , f: \homt_T( s, \sci t), \theta: \Theta(t) \vdash D(f, \theta) \ \type \\ 
\Gamma, s:  T^\core, \theta: \Theta(s) \vdash d(\theta) : D(1_s, \theta) 
}
{\Gamma, s: T^\op, t: T^\core , f: \homt_T( s,\sci t) , \theta: \Theta(t) \vdash \scel(d,f,\theta) : D(f,\theta)} 
\end{align*}
\begin{align*}
& \textsc{right } \homt \textsc{ computation}: \\
& \quad \inferrule
{ \Gamma \vdash T \ \type \\  
\Gamma, s: T^\core \vdash \Theta(s) \ \type \\ 
\Gamma, s: T^\core, t: T , f: \homt_T(\sci^\op s, t), \theta: \Theta(s) \vdash D(f,\theta) \ \type \\ 
\Gamma, s:  T^\core, \theta : \Theta(s) \vdash d(\theta) : D(1_s, \theta) 
}{
\Gamma, s:  T^\core, \theta : \Theta(s)  \vdash \scer(d,1_s,\theta ) \equiv d(\theta) : D(1_s,\theta)} 
\end{align*}
\begin{align*}
& \textsc{left }\homt \textsc{ computation}: \\
& \quad
 \inferrule
{ \Gamma \vdash T \ \type \\  
\Gamma, s: T^\core  \vdash \Theta(s) \ \type\\
\Gamma, s: T^\op, t: T^\core , f: \homt_T( s, \sci t), \theta: \Theta(t) \vdash D(f, \theta) \ \type \\ 
\Gamma, s:  T^\core, \theta: \Theta(s) \vdash d(\theta) : D(1_s, \theta) 
}
{\Gamma, s:  T^\core,  \theta: \Theta(s)  \vdash \scel(d,1_s,\theta) \equiv d(\theta) : D(1_s,\theta )} 
\end{align*}

\subsection{Transportation and composition}

Before we develop semantics for these rules, we prove some illustrative results in the syntax. 

First, we show that we can transport types along homomorphisms, which is analogous to the transportation along identities in homotopy type theory.

\begin{prop}\label{transport}
For any type $\Gamma, t:T \vdash S(t) \ \type$, we have a dependent term as follows.
\[\Gamma, t:T^\core, t':T, f: \homt_T(\sci^\op t,t'), s:S(\sci t) \vdash \mathsf{transport_R}(s,f):S(t') \]
\end{prop} 
\begin{proof}
By plugging the data
\begin{itemize}
\item $\Gamma, \vdash T \ \type$
\item $\Gamma, t: T^\core \vdash s:S(\sci t)$
\item $\Gamma, t:T^\core, t':T, f: \homt_T(\sci^\op t,t'), s:S(\sci t) \vdash S(t') \ \type$
\item $\Gamma, t: T^\core, s:S(\sci t)  \vdash s:S(\sci t) $
\end{itemize}
into the right $\homt$ computation rule, we find a term 
\[\Gamma, t:T^\core, t':T, f: \homt_T(\sci^\op t,t'), s:S(\sci t) \vdash \scer(\lambda s. s, f,s):S(t') \]
which we abbreviate by $\mathsf{transport_R}(s,f)$ (where that $\lambda s. s$ is merely abbreviation for the dependent term $\Gamma, t: T^\core, s:S(\sci t)  \vdash s:S(\sci t) $, and not a term constructor in our language.)
\end{proof}

Now using this transportation, we can show that we can compose two homomorphisms.

\begin{prop}\label{comp}
For any type $\vdash T \ \type$, we have a dependent term as follows.
\[r:T^\op, s:T^\core,t:T, f: \homt_T(r,\sci s), g: \homt_T(\sci^\op s, t)
\vdash \mathsf{comp_R}(f,g):\homt_T(r,t) \]
\end{prop}
\begin{proof}
Using Proposition \ref{transport}, we obtain a dependent term
\[r:T^\op, s:T^\core,t:T, f: \homt_T(r,\sci s), g: \homt_T(\sci^\op s, t)
\vdash \mathsf{transport_R}(f,g):\homt_T(r,t) \]
which we rename $\mathsf{comp_R}(f,g)$.
\end{proof}

By the right $\homt$ computation rule, we find that for any $f: \homt_T(r,\sci s)$, the composition $\mathsf{comp_R}(f,1_s) $ is strictly equal to $ f: \homt_T(r,\sci s)$. Using left $\homt$ elimination, we could analogously construct a left transportation $\mathsf{transport_L}$ and a left composition $\mathsf{comp_L}$ such that for any $g: \homt_T(\sci^\op s,t)$, the composition $\mathsf{comp_L}(1_s,g)$ is strictly equal to $g$ in the type $\homt_T(\sci^\op s,t)$.

\section{The semantics in $\Cat$}
\label{sec:semantics}

In this section, we construct an interpretation in the category of categories.

We work in a set theory with two uncountable inaccessible cardinals $\kappa < \lambda$. Let $\Cat$ denote the category of $\kappa$-small categories, and let $\CAT$ denote the $2$-category of $\lambda$-small categories. We have arranged this so that $\CAT$ contains $\Cat$ and all the objects of $\Cat$. Note that there are no universes in our syntax, so this is as much of a tower of universes in the semantics as we need.

\subsection{Grothendieck (op)fibrations}

The notions of Grothendieck fibrations and opfibrations play a central role in this interpretation, so we briefly review it here. See \cite{elephant} for more details.

Consider a functor $P: \mathcal E \to \mathcal B$ in $\CAT$. We say that a morphism $e: E' \to E$ in $\mathcal E$ is \emph{cartesian} if for any other morphism $e': E'' \to E$ in $\mathcal E$, any factorization of $P(e')$ through $P(e)$ (that is, a morphism $b: P(E'') \to P(E')$ such that $P(e) \circ b = P(e')$) uniquely lifts to a factorization of $e'$ through $e$ (that is, there is a unique morphism $\ell: E'' \to E'$ such that $e \circ \ell = e'$ and $P(\ell) = b$).
\[ \diagram 
\mathcal E \ar[d]^-P & E'' \ar@{-->}[r]_{\exists ! \ell} \ar@/^7pt/[rr]^{e'} & E' \ar[r]_-e & E \\
 \mathcal B & P(E'') \ar[r]^{b} \ar@/_7pt/[rr]_{P(e')}& P(E') \ar[r]^-{P(e)} & E 
\enddiagram \]

The archetypical cartesian morphisms are pullback squares, in the following sense. First, let $\two$ denote the poset $0 \leq 1$. Then $\C^\two$ is the category of morphisms of a category $\C$. Let $\cod :\C^\two \to \C$ be the functor which sends each morphism to its codomain. Now consider any category $\C$ in $\CAT$ with all pullbacks. Then a morphism in $\C^\two$ whose underlying diagram in $\C$ is a pullback square is a cartesian morphism with respect to the functor $\cod$.
 
A functor $P: \mathcal E \to \mathcal B$ in $\CAT$ is a \emph{Grothendieck fibration} if for any object $E$ of $\E$, any morphism $b: B \to P(E)$ in $\mathcal B$ can be lifted to a cartesian morphism $e: E' \to E$ (so that $P(e) = b$).

The archetypical Grothendieck fibration is a functor of the form $\cod :\C^\two \to \C$ for any category $\C$ with all pullbacks. This is because for any object $f$ in $\C^\two$ (which is a morphism $f: X \to Y$ in $\C$) and  morphism $g: Z \to \cod(f)$, the pullback square $\pi: g^*f \to f$ is a cartesian morphism in $\C^\two$ which is a lift of $g$.

Dually, we can define cocartesian morphisms and Grothendieck opfibrations. For a functor $P: \mathcal E \to \mathcal B$ in $\CAT$, a cocartesian morphism in $\mathcal E$ is a cartesian morphism in $\mathcal E^\ops$ with respect to the functor $P^\ops$. Then our archetypical cocartesian morphisms are pushout squares in any $\C^\two$ with respect to the functor $\dom :\C^\two \to \C$. A \emph{Grothendieck opfibration} is a functor $P: \mathcal E \to \mathcal B$ in $\CAT$ such that $P^\ops$ is a Grothendieck fibration. Our archetypical Grothendieck opfibration is $\dom :\C^\two \to \C$ for any category $\C$ with all pushouts.

One reason that Grothendieck fibrations have been of interest is that they are equivalent to pseudofunctors into $\Cat$. However, we will not make full use of that equivalence here, but we will make use of the \emph{Grothendieck construction}. Consider any category $\C$ in $\Cat$, and any functor $F: \C \to \Cat$. We can construct, following Grothendieck, a new category $\C.F$ as follows. The objects of $\C.F$ are pairs $(X, Y)$ such that $X$ is an object of $\C$ and $Y$ is an object of $F(X)$. Morphisms $(X, Y) \to (X',Y')$ in $\C.F$ are pairs $(f,g)$ such that $f: X \to X'$ is a morphism in $\C$ and $g: F(f)(Y) \to Y'$ is a morphism in $F(X')$. Then there is a projection functor $\pi_\C: \C.F \to \C$ which takes each pair $(X, Y)$ to $X$. The functor $\pi_\C$ is a Grothendieck opfibration since given any $(X,Y)$ in $\C.F$, we can lift any morphism $f: X \to X'$ in $\C$ to the cocartesian morphism $(f,1_{F(Y)}):(X,Y) \to (X', F(Y))$ in $\C.F$. 

Dually, given any functor of the form $F: \C^\op \to \Cat$, one can construct a Grothendieck fibration $\pi_\C: \C.F \to \C$. The objects of $\C.F$ are still pairs $(X, Y)$ such that $X$ is an object of $\C$ and $Y$ is an object of $F(X)$, but the morphisms $(X, Y) \to (X',Y')$ of $\C.F$ are pairs $(f,g)$ such that $f: X \to X'$ is a morphism in $\C$ and $g: Y \to F(f)(Y')$ is a morphism in $F(X)$. 

\subsection{The comprehension category}

A comprehension category over $\Cat$ (following \cite{Jac93}) consists of a category $\E$, a functor $\chi: \E \to \Cat^\two$, and a Grothendieck fibration $G: \E \to \Cat$ in $\CAT$ making the following diagram commute and such that $\chi$ sends to cartesian morphisms in $\mathcal E$ to cartesian morphisms in $\Cat^\two$ (that is, pullback squares).
\[ \diagram \E \ar[rr]^\chi \ar[rd]_G& & \Cat^\two \ar[dl]^\cod \\
& \Cat
\enddiagram \]

When interpreting a type theory into this framework, we take the objects of $\Cat$ to be our collection of contexts, the morphisms of $\Cat$ to be context morphisms, and each fiber $G^{-1} (\Gamma)$ to be the collection of types dependent on the context $\Gamma$. The terminal category $*$ is the interpretation of the empty context.
By applying the functor $\chi$ to an object $T$ of $G^{-1} (\Gamma)$, we obtain a morphism $\chi(T)$ in $\Cat$ with codomain $\Gamma$; the domain of $\chi(T)$ is the context extension (or comprehension) of $\Gamma$ by $T$. Thus, terms in a type $T$ dependent on a context $\Gamma$ are interpreted by sections of $\chi(T)$. 

\begin{prob}\label{compcat}
To construct a comprehension category over $\Cat$.
\end{prob}

\begin{example}
Before we construct the comprehension category, we indicate what $G^{-1} (*)$ (the types in the empty context) will be. This will be the category of functors $* \to \Cat$, which is isomorphic to $\Cat$ itself. The comprehension functor $\chi$ will take a functor of the form $\C: * \to \Cat$ to a functor which is isomorphic to the unique functor $\C(*) \to *$.
Thus, terms of $\C$ are functors $* \to \C(*)$, which are just objects of the category $\C(*)$.
\end{example}

\begin{construction}{compcat}
Consider the functor 
\[[-,\Cat]: \Cat^\ops \to \Cat\]
 which takes each object $\Gamma$ in $\Cat$ to the category $[\Gamma,\Cat]$ of functors $\Gamma \to \Cat$ in $\CAT$ and natural transformations between them. Performing the Grothendieck construction on this functor, we get the Grothendieck fibration 
 \[ \pi_\Cat:  \Cat.[-,\Cat ] \to \Cat. \]
The objects of $\Cat. [-,\Cat ] $ are pairs $(\Gamma, T )$ where $\Gamma$ is a category in $\Cat$ and $T: \Gamma \to \Cat $ is a functor. The morphisms $(\Gamma, T) \to (\Gamma', T')$ are pairs $(F,\tau)$ such that $F: \Gamma \to \Gamma'$ is a functor and $\tau: T \Rightarrow T' \circ F$ is a natural transformation. 
\[ \diagram 
{\ar@{=>}^\tau (5,-2)*{}; (5,-5)*{}}
\Gamma \ar[r]^T \ar[d]_F  & \Cat \\
\Gamma '  \ar[ur]_{T'}
\enddiagram \]

Now we define the functor $ \chi : \Cat. [-,\Cat ] \to \Cat^\two$ to be the Grothendieck construction. That is, for any object $(\Gamma, T)$ in $\Cat. [-,\Cat ]$, we have \[\chi(\Gamma, T)\ \ := \ \  \pi_\Gamma: \Gamma.T \to \Gamma.\] 
For any morphism $(F,\tau): (\Gamma, T) \to (\Gamma', T')$ we have that $\chi(F,\tau)$ is the morphism in $\C^\two$ with the following underlying commutative square in $\Cat$
\[ \diagram
\Gamma.T \ar[d]^{\pi_\Gamma} \ar[rr]^{\dom \chi(F,\tau)  }&& \Gamma'.T' \ar[d]^{\pi_{\Gamma'}} \\
\Gamma \ar[rr]^F && \Gamma'
\enddiagram \]
where $\dom \chi(F,\tau)$ is the functor which sends each $(X,Y)$ in $\Gamma.T$ to $(F(X), \tau_X(Y))$ in $\Gamma'.T'$.

By Lemma \ref{cartesianchi} below, the functor $\chi$ preserves cartesian morphisms.

Now it is straight-forward to check that the following diagram commutes.
\[ \diagram \Cat.[-,\Cat] \ar[rr]^\chi \ar[rd]_{\pi_\Cat}& & \Cat^\two \ar[dl]^\cod \\
& \Cat
\enddiagram \]
In particular, on objects, we have $\cod \chi (\Gamma, T) = \cod(\pi_\Gamma: \Gamma.T \to \Gamma) = \Gamma$
and $\pi_\Cat(\Gamma, T) = \Gamma$.
\end{construction}

\begin{lemma} \label{cartesianchi}
The functor $\chi$ preserves cartesian morphisms.
\end{lemma}
\begin{proof}
Consider a cartesian morphism $(F, \tau): (\Gamma, T) \to (\Gamma, T')$ in $\Cat.[-,\Cat]$.
\[ \diagram 
{\ar@{=>}^\tau (5,-2)*{}; (5,-5)*{}}
\Gamma \ar[r]^T \ar[d]_F  & \Cat \\
\Gamma '  \ar[ur]_{T'}
\enddiagram \]
Applying $\chi$ to this functor, we obtain the following morphism in $\Cat^\two$, which is a commutative square in $\Cat$. 
\[ \diagram
\Gamma.T \ar[d]^{\pi_\Gamma} \ar[rr]^{\dom \chi(F,\tau)  }&& \Gamma'.T' \ar[d]^{\pi_{\Gamma'}} \\
\Gamma \ar[rr]^F && \Gamma'
\enddiagram \]
We want to show that this is cartesian with respect to $\cod: \C^\two \to \C$.

So consider a morphism $\langle \alpha, \beta \rangle: Z \to \pi_{\Gamma'}$ in $\Cat^\two$ and a morphism $\gamma: Z_1 \to \Gamma$ in $\Cat$ such that $F \gamma = \beta$ as illustrated below.
\[ \diagram
Z_0 \ar@{-->}[r]^\delta \ar[d]^Z  \ar@/^20pt/[rrr]^\alpha & \Gamma.T \ar[d]^{\pi_\Gamma} \ar[rr]^{\dom \chi(F,\tau)  }&& \Gamma'.T' \ar[d]^{\pi_{\Gamma'}} \\
Z_1 \ar@/_20pt/[rrr]_\beta \ar[r]^{\gamma} & \Gamma \ar[rr]^F && \Gamma'
\enddiagram \tag{$*$} \]
We need to show that there is a unique $\delta: Z_0 \to \Gamma.T$ which makes the diagram commute.

Consider the object $(Z_0, c_{*,Z_0})$ in $\Cat.[-,\Cat]$ (where $ c_{*,Z_0}$ is the functor $Z_0 \to \Cat$ which sends everything to the terminal category).
Then there is a morphism $(\beta Z, \overline \alpha):
(Z_0, c_{*,Z_0}) \to (\Gamma', T')$ where $\overline \alpha_z$ is the functor $* \to T'F\beta Z(z)$ which picks out the object $\alpha(z)$ for any object $z \in Z_0$.
\[ \xymatrix@C=6em{
{\ar@{=>}^\tau (12,-16)*{}; (12,-20)*{}}
{\ar@{=>}^{\overline \delta} (12,-6)*{}; (12,-10)*{}}
{\ar@{=>}^<<<<{\overline \alpha} (5,-5)*{}; (5,-21)*{}}
Z_0 \ar@/^/[dr]^{c_{*,Z_0}} \ar[d]_{\beta Z}\\
\Gamma \ar[r]|<<<<\hole^>>>>>>>T \ar[d]_F  & \Cat \\
\Gamma '  \ar@/_/[ur]_{T'}
} \]
Since $(F,\tau)$ is cartesian, there is a unique natural transformation $\overline \delta: c_{*,Z_0} \Rightarrow T\beta Z$ such that $\overline \alpha = \tau \overline \delta$. This gives rise to the functor $\delta: Z_0 \to \Gamma.T$ by setting $\delta(z) := (\Gamma, \overline \delta_z(*))$. This makes the diagram ($*$) above commute. 

If there were another such functor $\delta'$ making the diagram ($*$) commute, by making an analogous argument, we would find a natural transformation $\overline \delta'$. But then $\overline \delta' = \overline \delta$, and so $\delta' = \delta$. Thus, $\delta$ is the unique functor making the diagram ($*$) commute.

Thus, $\chi(F, \tau)$ is cartesian for every cartesian $(F, \tau)$ so $\chi$ is cartesian.
\end{proof}

This comprehension category has a unit $\eta: \Cat \to \Cat.[-,\Cat]$ which takes a category $\Gamma$ to $(\Gamma, *_\Gamma: \Gamma \to \Cat)$ where $*_\Gamma$ is the functor which is constant at the terminal category $*$. Then $\chi (\Gamma, *_\Gamma)$ is an isomorphism $\pi_\Gamma : \Gamma. *_\Gamma \to \Gamma$. Recall that for any $(\Gamma, T)$ in $\Cat.[-,\Cat]$, we interpret the terms of $T$ as sections of $\chi(\Gamma, T)$. These are in bijection with functors $ \Gamma. *_\Gamma \to \dom \chi(\Gamma, T)$ over $\Gamma$ (where the functor $\dom$ sends a functor to its domain). These are in bijection with morphisms $(\Gamma, *_\Gamma) \to (\Gamma, T)$ in the fiber $\pi_\Cat^{-1} \Gamma$ which are just morphisms $*_\Gamma  \to  T$ in the category $[\Gamma, \Cat]$. 

In the sections that follow, we will interpret contexts as categories $\Gamma$ in $\Cat$, we will interpret types in context $\Gamma$ as objects of the category $[\Gamma, \Cat]$ (which are functors $\Gamma \to \Cat$), and we will interpret terms of such an object $T$ in $[\Gamma, \Cat]$ as morphisms from $*_\Gamma \to T$ in $[\Gamma, \Cat]$ (which are natural transformations $*_\Gamma \Rightarrow T$). 

\subsection{\textsf{core} and \textsf{op}}

In this section, we interpret the type formers $\core$ and $\op$ and the functions $\sci$ and $\sci^\op$ into our comprehension category. 

\begin{prob}\label{coreop}
To construct functorially, from each functor $\T: \Gamma \to \Cat$, functors $\T^\cores: \Gamma \to \Cat$ and $\T^\ops: \Gamma \to \Cat$, natural transformations $i: \T^\cores \Rightarrow T$ and $i^\ops : \T^\cores \Rightarrow \T^\ops$.
\end{prob}

\begin{example}
Before we give the construction, we describe what it will be in the empty context.

Consider a type $\C$ in the empty context, which is just a category. Then $\C^\cores$ will be just the set\footnote{In this note, we will often talk about a \emph{set} $X$ in the category $\Cat$. To be clear, we mean the category whose objects are the elements of $X$ and whose morphisms are only identity morphisms.} of objects of $\C$ and $\C^\ops$ will be the usual opposite category of $\C$. Then $i: \C^\cores \to \C$ and $i^\ops: \C^\cores \to \C^\ops$ will be the injections which are the identity on objects.
\end{example}

\begin{construction}{coreop}
Recall the two oft-encountered endofunctors $(-)^\cores$ and $(-)^\ops$ on $\Cat$. For any category $\C$, the category  $\C^\cores$ is the set $\ob(\C)$.
The objects of $\C^\ops$ are the objects of $\C$, and for all objects $X$, $Y$ of $\mathcal C^\ops$, the homomorphisms $X \to Y$ in $\C^\ops$ are the homomorphisms $Y \to X$ in $\C$. There are moreover natural transformations $i: (-)^\cores \Rightarrow id_\Cat $ and $i^\op: (-)^\cores \Rightarrow (-)^\ops $ induced by the identity function on objects. 

Postcomposition with the functors $(-)^\cores$ and $(-)^\ops$ gives two functors which take any $\T : \Gamma \to \Cat$ to $\T^{\cores}: \Gamma \to \Cat$ and $\T^{\ops}: \Gamma \to \Cat$, respectively. 

Whiskering with $i$ (resp. $i^\ops$) gives a functor which takes any $\T : \Gamma \to \Cat$ to the natural transformation $i: \T^\cores \Rightarrow \T$ (resp. $i^\ops: \T^\cores \Rightarrow \T^\ops$).
\end{construction}

\subsection{$\homt$ {formation}}
In this section, we interpret the $\homt$ {formation} rule.
\begin{prob}\label{formation}
To functorially construct for each functor $\T: \Gamma \to \Cat$ and natural transformations $s: *_\Gamma \Rightarrow \T^\op , t: *_\Gamma \Rightarrow \T$, a functor $\homs_\T(s,t): \Gamma \to \Cat$.
\end{prob}

\begin{example}
In the empty context, our problem reduces to finding for each category $\C$ and object $(x,y) \in \C^\ops \times \C$, a category $\homs_\C(x,y)$. Our construction will produce the usual homset $\homs_\C(x,y)$.
\end{example}

\begin{construction}{formation}
Let $(-)^\ops \times id_\Cat: \Cat \to \CAT$ denote the functor which takes any $\C \in \Cat$ to $\C^\op \times \C$ in $\CAT$. Then let $c_\Set: \Cat \to \CAT$ be the constant functor at $\Set$, the category of $\kappa$-small sets. 

There is a lax natural transformation $\homs: (-)^\ops \times id_\Cat \Rightarrow c_\Set$ which at each $\C \in \Cat$ is the usual functor $\homs_\C: \C^\ops \times \C \Rightarrow \Set$ (which maps an object $(X,Y) \in \C^\ops \times \C$ to the set $\homs_\C(X,Y)$ of morphisms $X \to Y$ in $\C$). This transformation is not natural since for a functor $F: \C \to \D$, the square of functors might not commute. It is however a \emph{lax} natural transformation since there is a natural transformation $\homs_F: \homs_\C \to \homs_\D(F^\ops \times F)$ which at each $(X, Y) \in \C^\ops \times \C$ is the function $ \homs_\C(X,Y) \to \homs_\D(F^\ops X , FY)$, the morphism part of the functor $F$.
\[ \diagram
\C^\op \times \C \ar[d]_{\homs_\C}^>>>>>>{\ \ \ \ \ \xLongrightarrow{\homs_F}} \ar[r]^{F^\op \times F}  & \D^\op \times \D \ar[d]^{\homs_\D} 
\\
 \Set \ar@{=}[r] & \Set 
\enddiagram
\]

Now for any functor $\T: \Gamma \to \Cat$ and natural transformations $s: *_\Gamma \Rightarrow \T^\op , t: *_\Gamma \Rightarrow \T$, we have the following diagram of functors and (lax) natural transformations.
\[ \xymatrix@=5em{
\Gamma \rruppertwocell<10>^{\ *_\Gamma \ }{(s,t)\ \ \ \ \ \ \ \ \ \  } \ar[r]_-\T &{ \Cat} \ar[r]^{(-)^\ops \times id \ \ } \rlowertwocell<-10>_{c_\Set}{\ \ \ \homs}&{ \CAT}\\
} \tag{$*$}\]

By pasting together the natural transformation $(s,t)$ and the lax natural transformation $\homs$ in the diagram above, we obtain a lax natural transformation $\homs_\T(s,t): *_\Gamma \Rightarrow c_\Set  \T$. 
At every object $\gamma \in \Gamma$, this gives a functor $\homs_\T(s,t)\gamma: * \to \Set$ which is just a set $\homs_\T(s,t)\gamma(*)$. For every morphism $\phi: \gamma \to \gamma'$ in $\Gamma$, this produces a natural transformation  
\[ \xymatrix@C=5em{
{*} \ar[d]_{\homs_\T(s,t)\gamma}^{\ \ \ \xLongrightarrow{\homs_\T(s,t)\phi}} \ar@{=}[r] & {*} \ar[d]^{\homs_\T(s,t)\gamma'} \\
\Set \ar@{=}[r]  & \Set
} \]
which is just a function $\homs_\T(s,t)\phi_*: \homs_\T(s,t)\gamma(*) \to \homs_\T(s,t)\gamma'(*)$. In other words, $\homs_\T(s,t)$ is a functor $\Gamma \to \Cat$.
\end{construction}
\subsection{$\homt$ introduction}

In this section, we interpret the $\homt$ introduction rule.

\begin{prob}\label{introduction}

To functorially construct for each functor $\T: \Gamma \to \Cat$ and natural transformation $ t: *_\Gamma \Rightarrow \T^\cores$, a natural transformation $1_t: *_\Gamma \Rightarrow \homs_\T(i^\ops t,i t)$.
\end{prob}
\begin{example}
In the empty context, our problem reduces to finding for each category $\C$ and object $t \in \C^\cores$, an object $1_t \in \hom_T( t, t)$. Our construction will produce the usual identity morphism $t \to t$.
\end{example}
\begin{construction}{introduction}
For every category $\C$ and for every $t \in \C^\cores$ (that is, an object $t$ of $\C$), there is an element $1^\C_t \in \homs_\C(t,t)$, the identity morphism $t \to t $. Since this is (vacuously) natural in $t$, it assembles into the natural transformation shown below
\[ \xymatrix@C=5em{
\C^\cores \rtwocell^{c_*}_{\homs_\C(i^\ops \times i)}{\ 1_\bullet^\C} & \Set
} \]
where $c_*$ is the constant functor at the terminal set. 

This is furthermore natural in $\C$, so it assembles into the following modification.
\[ \xymatrix@C=10em{ 
\Cat \ar@/^20pt/[rr]^{(-)^\cores}\ar@/_20pt/[rr]_{c_\Set} & & \CAT
{\ar@{=}@/^-15pt/_{c_*} (36,5)*{}; (36,-5)*{}}
{\ar@{=>}@/^-15pt/ (36,5)*{}; (37,-5.7)*{}}
{\ar@{=}@/^15pt/^{\homs_\C(i^\op \times i)} (45,5)*{}; (45,-5)*{}}
{\ar@{=>}@/^15pt/ (45,5)*{}; (44,-5.7)*{}}
{\ar@3^{1_\bullet} (37,0)*{}; (46,0)*{}}
} \]

Now, consider a functor $\T: \Gamma \to \Cat$ and a natural transformation $ t: *_\Gamma \Rightarrow \T^\cores$. We obtain a diagram 
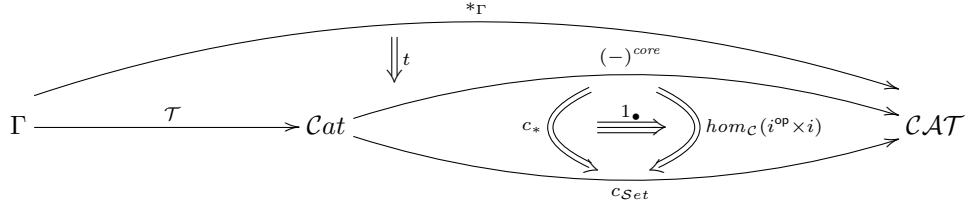
\[ \xymatrix@C=10em{ \Gamma \ar[r]^\T \ar@<10pt>@/^30pt/[rrr]^{*_\Gamma} & 
\Cat \ar@/^20pt/[rr]^{(-)^\cores}\ar@/_20pt/[rr]_{c_\Set} & & \CAT
{\ar@{=}@/^-15pt/_{c_*} (76,5)*{}; (76,-5)*{}}
{\ar@{=>}@/^-15pt/ (76,5)*{}; (77,-5.7)*{}}
{\ar@{=}@/^15pt/^{\homs_\C(i^\op \times i)} (85,5)*{}; (85,-5)*{}}
{\ar@{=>}@/^15pt/ (85,5)*{}; (84,-5.7)*{}}
{\ar@{=>}^t (50,12)*{}; (50,6)*{}}
{\ar@3^{1_\bullet} (77,0)*{}; (86,0)*{}}
} \]
Then pasting together this diagram, we obtain 
a modification $c_* \Rrightarrow \hom_\C(i^\ops t, i t)$ which unfolds into a natural transformation $1_t: * \Rightarrow \hom(\iota^\op t,\iota t)$. This natural transformation has at each $\gamma \in \Gamma$ the component $1_{t \gamma}: * \to \hom( t \gamma, t \gamma)$ which picks out the identity morphism $t \gamma \to t \gamma$.
\end{construction}

\subsection{\textsc{Right} $\homt$ {elimination} and {computation}}

In this section, we interpret the right $\homt$ elimination rule so that it satisfies the strict equality required by the right $\homt$ computation rule.

\begin{prob}\label{rightelim}
To functorially construct for 
each functor $\T: \Gamma \to \Cat$, each functor $\Theta: \Gamma.\T^\cores \to \Cat$,
each functor $\D: \Gamma.\T^\cores . \T. \homs_\T (i^\ops \times 1).\Theta \to \Cat$, and 
each natural transformation 
\[d: *_{\Gamma.\T^\core.\Theta} \Rightarrow \D\circ 1_\bullet.\Theta:\Gamma.\T^\core.\Theta \to \Cat,\] a natural transformation 
\[e_r(d): *_{ \Gamma.\T^\cores . \T. \homs_\T (i^\ops \times 1).\Theta} \Rightarrow \D:  \Gamma.\T^\cores . \T. \homs_\T (i^\ops \times 1).\Theta \to \Cat,\] such that $e_r(d)\circ 1_\bullet.\Theta = d$.
\end{prob}

\begin{example}
Where $\Gamma$ and $\Theta$ are both empty contexts, our problem reduces to finding for 
each category $\C$, 
each functor 
$\D: \C^\cores.\C.\homs_\C(i^\ops \times 1) \to \Cat$,
and each natural transformation 
$d: *_{\C^\cores} \Rightarrow \D\circ 1_\bullet: \C^\core \to \Cat,$ a natural transformation 
\[e_r(d): *_{\C^\cores.\C.\homs_\C(i^\ops \times 1)} \Rightarrow \D:   \C^\cores.\C.\homs_\C(i^\ops \times 1) \to \Cat,\] 
such that $e_r(d) 1_\bullet = d$.

The objects of $\C^\cores.\C.\homs_\C(i^\ops \times 1)$ are triples $(X,Y,f)$ such that $f: X \to Y$ is a morphism in $\C$. Its morphisms are of the form $g: (X,Y,f) \to (X,Y',g \circ f)$ where $g:Y \to Y'$ is a morphism of $\C$.  

Since $e_r(d) 1_\bullet$ must be $d$, we know that $e_r(d)_{(X,X,1_X)}$ must be $d(X)$ for any object $X$ of $\C^\cores$.
Now for any other object $(X,Y,f)$ of $\C^\cores.\C.\homs_\C(i^\ops \times 1)$, there is a morphism $f: (X,X,1_X) \to (X,Y,f)$.
Then we get a morphism $\D(f) : \D(X,X,1_X) \to \D(X,Y,f)$. We apply this functor to $d(X)$ to obtain $e_r(d)_f$. That is, we set
\[ e_r(d)_f := \D(f) \left( dX \right).\]

\begin{construction}{rightelim}
Consider the data given in the Problem.
The objects of the category $\Gamma.\T^\cores . \T. \homs_\T (i^\ops \times 1).\Theta$ are quintuples $(\gamma, s,t, f, \theta)$ where $\gamma$ is an object of $ \Gamma$, $f: s \to t$ is a morphism in the category $\T(\gamma)$, and $\theta$ is an object of $ \Theta(\gamma, s)$.
Morphisms in $\Gamma.\T^\cores . \T. \homs_\T (i^\ops \times 1).\Theta$ are of the form 
\[(\phi, \alpha, \beta) : (\gamma, s,t, f, \theta) \to (\gamma', \ \T(\phi)(s) ,\ t', \ \alpha \circ \T(\phi)(f),\ \theta') \tag{$*$}\]
where $\phi: \gamma \to \gamma'$ is a morphism in $\Gamma$, $\alpha : \T(\phi)(t) \to t'$ is a morphism in $\T(\gamma')$, and $\beta: \Theta(\phi)(\theta) \to \theta'$ is a morphism in $\Theta(\gamma', \T(\phi)(s))$.

The objects of the category $\Gamma.\T^\cores.\Theta$ are triples $(\gamma, t, \theta)$ where $\gamma \in \Gamma$, $t \in \T^\cores(\gamma)$, and $\theta \in \Theta(\gamma,t)$. The morphisms are of the form $(\phi,\beta): (\gamma, t,\theta) \to (\gamma', T(\phi)(t) ,\theta')$ where $\phi: \gamma \to \gamma'$ is a morphism in $\Gamma$ and $\beta: \Theta(\phi)\theta \to \theta'$ is a morphism in $\Theta(\gamma', T(\phi)(t))$.

The functor $1_\bullet.\Theta: \Gamma.T^\core.\Theta \to \Gamma.\T^\cores . \T. \homs_\T (i^\ops \times 1).\Theta$ takes a triple $(\gamma, t, \theta)$ to the quintuple $(\gamma, t,t, 1_t, \theta)$.

Now, to construct the natural transformation
\[e_r(d): * \Rightarrow \D:  \Gamma.\T^\cores . \T. \homs_\T (i^\ops \times 1).\Theta \to \Cat,\]
we need to first specify all of its components $e_r(d)_{(\gamma, s,t, f,\theta)}: * \to \D(\gamma, s,t, f,\theta)$.

So fix an object $(\gamma, s,t, f,\theta)$ of the category $ \Gamma.\T^\cores . \T. \homs_\T (i^\ops \times 1).\Theta$. To construct a functor 
$e_r(d)_{(\gamma, s,t, f,\theta)}: * \to \D(\gamma, s,t, f,\theta)$ is to find an object of the category $\D(\gamma, s,t, f,\theta)$.
Consider the morphism $(1_\gamma, f,1_\theta): (\gamma, s,s, 1_s,\theta) \to (\gamma, s,t, f,\theta)$ in $ \Gamma.\T^\cores . \T. \homs_\T (i^\ops \times 1).\Theta$. By applying the functor $\D$, we get a functor  
\[\D(1_\gamma, f,1_\theta): \D\circ 1_\bullet (s).\Theta = \D(\gamma, s,s, 1_s,\theta) \to \D (\gamma, s,t, f,\theta). \]
Now we apply this functor to the object $d(\gamma,s,\theta)(*)$ of $\D\circ 1_\bullet (\gamma, s,\theta)$ to get $e_r(d)_{(\gamma, s,t, f, \theta)} $ as follows:
\[ e_r(d)_{(\gamma, s,t, f, \theta)}(*) := \D(1_\gamma,f,1_\theta) d(\gamma, s,\theta)(*). \]

We need to check that this choice is natural, which amounts to showing that for any morphism in $ \Gamma.\T^\cores . \T. \homs_\T (i^\ops \times 1).\Theta$ as in line ($*$) above, we have the equality
\[ \D(\phi, \alpha,\beta) e_r(d)_{(\gamma, s,t, f,\theta)} =   e_r(d)_{(\gamma', T(\phi)( s) ,t', \alpha  T(\phi)(f),\theta')} .\]
But we have 
 \begin{align*}
 \D (\phi, \alpha,\beta) e(d)_{(\gamma, s,t, f,\theta)} &=  \D (\phi, \alpha,\beta) \D (1_\gamma,f,1_\theta)d(\gamma, s,\theta) \\
&= \D (\phi, \alpha T(\phi)(f), \beta)d(\gamma, s, \theta)\\
&= \D (1_{\gamma'},\alpha  T(\phi)(f),1_\theta')  \D (\phi, 1_{T(\phi)s},\beta) d(\gamma,s,\theta)\\
 &= \D (1_{\gamma'},\alpha  T(\phi)(f),1_\theta')d(\gamma', T(\phi)(s),\theta') \\
&=  e_r(d)_{(\gamma', T(\phi)( s) ,t', \alpha  T(\phi)(f),\theta')}
\end{align*} 
where the first equality and last equality come from our definition of $e_r(d)$, the second and third from composition in the category $ \Gamma.\T^\cores . \T. \homs_\T (i^\ops \times 1).\Theta$, and the forth from the naturality of $d$.

Now we only need to check that $e_r(d)\circ 1_\bullet.\Theta = d$. For any $(\gamma, t,\theta)$ of $\Gamma.T^\core.\Theta$, we have that 
\begin{align*}
e_r(d) {1_\bullet}.\Theta {(\gamma,t.\theta)} &= e_r(d) {(\gamma, t,t,1_t,\theta)} \\
&= \D (1_\gamma, 1_t,1_\theta)d(\gamma, t,\theta) \\
&= d(\gamma,t,\theta)
\end{align*} 
where the first equality is the definition of $1_\bullet$, the second is the definition of $e_r(d)$, and the third is the functoriality of $\D $.
Therefore, we have that $e_r(d) 1_\bullet = d$.
\end{construction}
\end{example}

\subsection{{Left} $\homt$ {elimination} and {computation}}

In this section, we interpret the left $\homt$ elimination rule so that it satisfies the strict equality required by the left $\homt$ computation rule. However, it just takes an application of right $\homt$ elimination and computation.

\begin{prob}\label{leftelim}
To functorially construct for 
each functor $\T: \Gamma \to \Cat$, each functor $\Theta: \Gamma.\T^\cores \to \Cat$,
each functor $\D: \Gamma.\T^\ops . \T^\cores. \homs_\T (1 \times i).\Theta \to \Cat$, and 
each natural transformation 
\[d: *_{\Gamma.\T^\core.\Theta} \Rightarrow \D\circ 1_\bullet.\Theta:\Gamma.\T^\core.\Theta \to \Cat,\] a natural transformation 
\[e_r(d): *_{ \Gamma.\T^\ops . \T^\cores. \homs_\T (1 \times i).\Theta} \Rightarrow \D:  \Gamma.\T^\ops . \T^\cores. \homs_\T (1 \times i).\Theta \to \Cat,\] such that $e_r(d)\circ 1_\bullet.\Theta = d$.
\end{prob}

\begin{construction}{leftelim}
Note that $(\T^\ops)^\cores = \T^\cores$,
and \[\Gamma.(\T^\ops)^\cores . \T^\ops. \homs_{\T^\ops} (i \times 1).\Theta \cong \Gamma.\T^\ops . \T^\cores. \homs_\T (1 \times i).\Theta\]
 and $(T^\ops)^\cores = T^\cores$. Then by substituting $T^\ops$ in for $T$ everywhere in the statement of the Problem \ref{rightelim}, we find that our problem has already been solved.
\end{construction}

\subsection{The homotopy theory}
In this section, we describe the (directed) homotopy theory that this interpretation inhabits in $\Cat$.

To understand this homotopy theory in $\Cat$, first recall that there are two weak factorization systems on $\Cat$ \cite{Nor}. We begin with the following two factorizations of the fold map
\[ \diagram
* + * \ar[r]^-{0 + 1}& \two \ar[r]^{!} & *
\enddiagram \tag{$\rightarrow$} \]\[
\diagram
* + * \ar[r]^-{0 + 1}& \mathbb I \ar[r]^{!} & *
\enddiagram \tag{$\cong$} \]
where $\two$ is the category generated by one morphism $0 \to 1$, $\mathbb I$ is the category generated by one isomorphism $0 \cong 1$. This gives us, by exponentiation, two functorial factorizations of the diagonal $\C \to \C \times \C$ of any category $\C$.
\[ 
\xymatrix@C=4em{
\C \ar[r]^-{\C^!} & \C^\two \ar[r]^-{\C^0 \times \C^1} & \C \times \C
} \tag{$\C^\rightarrow$}\]
\[ 
\xymatrix@C=4em{
\C \ar[r]^-{\C^!} & \C^{\mathbb I} \ar[r]^-{\C^0 \times \C^1} & \C \times \C
} \tag{$\C^{\cong}$}\]
Then, we can functorially factor any functor $F: \C \to \D$ in two ways.
\[ 
\xymatrix@C=4em{
\C \ar[r]^-{1 \times \D^! F} & \C \times_{\D^0} \D^\two \ar[r]^-{\D^1} & \D
} \tag{$\ell^\rightarrow, r^\rightarrow$}\]
\[ 
\xymatrix@C=4em{
\C \ar[r]^-{1 \times \D^! F} & \C \times_{\D^0} \D^{\mathbb I} \ar[r]^-{\D^1} & \D 
} \tag{$\ell^{\cong}, r^{\cong}$} \]
This gives the factorizations of two functorial weak factorization systems $(\mathcal L^\rightarrow, \mathcal R^\rightarrow)$ and $(\mathcal L^{\cong}, \mathcal R^{\cong})$ on $\Cat$.

We can characterize the right maps of these weak factorization systems in another way. The maps in $\mathcal R^\rightarrow$ are exactly those maps which have the enriched right lifting property against the inclusion of the domain into the generic morphism: $0:* \to \two$. The maps in $\mathcal R^{\cong}$ are exactly those maps which have the enriched right lifting property against the inclusion of the domain into the generic isomorphism: $0: * \to \mathbb I$.

While the usual identity type of Martin-L\"of type theory has an interpretation into the weak factorization system $(\mathcal L^{\cong}, \mathcal R^{\cong})$ \cite{Nor}, we claim that we have constructed in the preceding sections an interpretation into the weak factorization system $(\mathcal L^\rightarrow, \mathcal R^\rightarrow)$. By this, we just mean that (1) given any interpretation of a dependent type $\T: \Gamma \to \Cat$, its comprehension $\pi_\Gamma: \Gamma .T \to \Gamma$ is in $\mathcal R^\rightarrow$, and (2) the elimination property can be seen as the existence of a lift between a map of $\mathcal L^\rightarrow$ and a map of $\mathcal R^\rightarrow$.

To show that any comprehension $\pi_\Gamma: \Gamma.\T \to \Gamma$ is in $\mathcal R^\rightarrow$, note first that it is a Grothendieck opfibration.
\begin{prop}
Any Grothendieck opfibration $p: E \to B$ is in $\mathcal R^\rightarrow$.
\end{prop}
\begin{proof}
We need to show that there is a lift $\ell$ in the following square.
\[ \diagram
E \ar[d]_{1 \times B^!p}\ar@{=}[r]& E \ar[d]^p \\
E \times_{B^0} B^\two \ar@{-->}[ur]^\ell \ar[r]_-{B^1} & B
\enddiagram \]

An object of $E \times_{B^0} B^\two$ is a pair $(e, f)$ where $e$ is an object of $E$ and $f:pe \to b$ is a morphism in $B$. By the definition of Grothendieck opfibration, there is a lift of $f$, which we will denote $\widetilde{(e,f)}: e \to \ell(e,f)$ such that $p \widetilde{(e,f)} = p$. We will choose $\widetilde{(e,1_e)}$ to always be $1_e$.

A morphism $(e_0, f_0) \to (e_1, f_1)$ of $E \times_{B^0} B^\two$ is a pair $(\epsilon, \phi)$ where $\epsilon: e_0 \to e_1$ is morphism of $E$ and $\phi: b_0 \to b_1$ is a morphism of $B$ which makes the following square commute in $B$.
\[ \diagram 
pe_0 \ar[d]^{f_0} \ar[r]^{p \epsilon} &pe_1  \ar[d]^{f_1}\\
b_0  \ar[r]^\phi & b_1
\enddiagram \]
Now, in $E$ we have the following diagram of solid arrows.
\[ \diagram 
e_0 \ar[d]_{\widetilde{(e_0,f_0)}} \ar[r]^{ \epsilon} &e_1  \ar[d]^{\widetilde{(e_1,f_1)}}\\
\ell(e_0,f_0) \ar@{-->}[r]  & \ell(e_1,f_1)
\enddiagram \]
By the definition of Grothendieck opfibration, we obtain a unique morphism $\ell(e_0,f_0) \to \ell(e_1,f_1)$, which we will denote $\ell(\epsilon, \phi)$, which makes the above diagram into a commutative square. 

It is straightforward to check that $\ell$ is a functor which makes the first diagram of this proof commute.
\end{proof}

Now the right elimination property (in the empty context) says that for any category $\C$, any functor $\D: \C^\core.\C.\homs_\C \to \Cat$, and any functor $d: \C^\core \to \C^\core.\C.\homs_\C.\D $ making the following square commute, there is a lift $\ell$ making both triangles commute.
\[ \diagram
\C^\core \ar[r]^-d \ar[d]^{1_\bullet} & \C^\core.\C.\homs_\C.\D \ar[d]^\pi \\
\C^\core.\C.\homs_\C \ar@{=}[r] \ar@{-->}[ur]^\ell & \C^\core.\C.\homs_\C
\enddiagram \]
We want to see that this arises as the solution of a lifting problem between a left map and a right map of our weak factorization system. We just saw that $\pi$ is a right map, so now we prove that $1_\bullet$ is a left map.

\begin{prop}
For any category $\C$, the functor 
\[ 1_\bullet : \C^\core \to \C^\core.\C.\homs_\C \]
is in $\mathcal L^\rightarrow$.
\end{prop}
\begin{proof}
Factor the functor $ i: \C^\core \to \C$. 
\[ \diagram
 \C^\core \ar[r]^-{1 \times \C^! i}   &\C^\core \times_{\C^0} \C^\two \ar[r]^-{\C^1} &\C 
 \enddiagram \]
 
 We claim that the middle object, $\C^\core \times_{\C^0} \C^\two$, is isomorphic to $\C^\core.\C.\homs_\C$ and this makes $1_\bullet$ isomorphic to $1 \times \C^! i$, the left factor of the factorization for the weak factorization system $(\mathcal L^\rightarrow, \mathcal R^\rightarrow)$. Thus, $1_\bullet$ will be in $\mathcal L^\rightarrow$.
 
 The objects of $\C^\core.\C.\homs_\C$ are pairs $((x,y),f)$ where $(x,y)$ is an object of $\C^\op \times \C$ and $f$ is morphism $x \to y$. 
 
 The objects of $\C^\core \times_{\C^0} \C^\two$ are pairs $(x, f)$ where $x$ is an object of $\C^\core$ (that is, $x$ is an object of $\C$) and $f$ is morphism with domain $x$.
 
On objects, define the functor $\alpha: \C^\core.\C.\homs_\C \to \C^\core \times_{\C^0} \C^\two$ by $\alpha((x,y),f) := (x,f)$ and its inverse by $\alpha^{-1} (x,f) := ((x, \cod f), f)$.

 The morphisms $((x,y),f) \to ((x,y'),f')$ of $\C^\core.\C.\homs_\C$ are morphisms $(1_x,g): (x,y) \to (x,y')$ of $\C^\op \times \C$ such that $gf = f'$.
 
      The morphisms $(x,f) \to (x,f')$ of $\C^\core \times_{\C^0} \C^\two$ are morphisms  $(1_x,g): f \to f'$ of $\C^\two$ (so that $gf = f'$).
      
      On morphisms, define the functor $\alpha$ and its inverse both by $(1_x,g) \mapsto (1_x,g)$.
      
      This gives an isomorphism $\alpha: \C^\core.\C.\homs_\C \cong \C^\core \times_{\C^0} \C^\two$.

Now, we just check that $\alpha 1_\bullet  = 1 \times C^! i$. For any object $x$ of $\C^\core$, we have that $\alpha 1_\bullet (x)$ and $1 \times C^! i(x)$ are both $(x,1_x)$ (and $\C^\core$ has no nonidentity morphisms).
 \end{proof}
 
 \section{Directed homotopy theory and concurrency} 
 \label{sec:dhott}
 In this section, we briefly describe some aspects of directed homotopy theory and its application to concurrency and indicate how our $\homt$ types can be intepreted in a category of directed spaces.
 
 \subsection{Background}
 In brief, directed homotopy theory is the study of spaces for which certain paths are singled out and given a direction \cite{Grandis}. For example, one might want to consider the circle $S^1$ in the category of topological spaces together with the information that any clockwise path is `allowed', but counter-clockwise paths are `not allowed'.

Such directed spaces appear in many applications, but perhaps most notably in the study of concurrent processes. A system of concurrent processes can be represented by a space whose points represent states of the system and whose directed paths represent possible executions of the system. For example, we might consider two processes $A$ and $B$ which both need to access the same two memory locations $m$ and $n$ in succession \cite{FGHMR}. In order to avoid a conflict, we allow each process to `lock' ($L$) and `unlock' ($U$) each memory location, and while one memory location is locked by one process, we disallow the other process to lock it. Then the set of allowed states in such a system can be seen to form a directed space $S$ (shaded in gray in Figure \ref{fig:swissflag}) whose underlying topological space is a subspace of the square $[0,1]^2$ and whose directed paths are those paths which are monotonically non-decreasing in both $A$ and $B$.
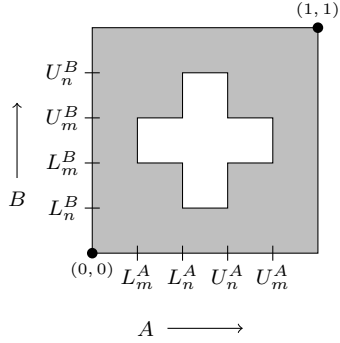
\begin{figure} 
\begin{tikzpicture}
\filldraw[fill=lightgray]
(0,0) -- (0,3) -- (3,3) -- (3,0) -- cycle 
(1.2,1.2) -- (1.2,0.6) -- (1.8,0.6) -- (1.8,1.2) -- (2.4,1.2) -- (2.4,1.8) -- (1.8,1.8) -- (1.8,2.4) -- (1.2,2.4) -- (1.2,1.8) -- (0.6,1.8) -- (0.6,1.8) -- (0.6,1.2) -- cycle;
  \draw[-] (0.6,-0.1) -- +(0,0.2) node[below=4pt] {\footnotesize $ L_m^A$};
    \draw[-] (-0.1,0.6) -- +(0.2,0) node[left=4pt] {\footnotesize $ L_n^B$};
    \draw[-] (1.2,-0.1) -- +(0,0.2) node[below=4pt] {\footnotesize $L_n^A$};
        \draw[-] (-0.1,1.2) -- +(0.2,0) node[left=4pt] {\footnotesize $ L_m^B$};
        \draw[-] (1.8,-0.1) -- +(0,0.2) node[below=4pt] {\footnotesize $U_n^A$};
            \draw[-] (-0.1,1.8) -- +(0.2,0) node[left=4pt] {\footnotesize $U_m^B$};
        \draw[-] (2.4,-0.1) -- +(0,0.2) node[below=4pt] {\footnotesize $U_m^A$};
                    \draw[-] (-0.1,2.4) -- +(0.2,0) node[left=4pt] {\footnotesize $U_n^B$};
\draw[->] (-1,1)  -- +(0,1) node[below=30pt] {\footnotesize $ B$};
\draw[->] (1,-1)  -- +(1,0) node[left=30pt] {\footnotesize $ A$};
\fill (0,0)  circle[radius=2pt] node[below] {\tiny $(0,0)$};
\fill (3,3)  circle[radius=2pt] node[above] {\tiny $(1,1)$};
  \end{tikzpicture}
  \caption{The Swiss flag}
    \label{fig:swissflag}
  \end{figure}

Now we can ask which states $(x,y)$ are \emph{reachable} by these two processes starting from the origin $(0,0)$ and which states $(x,y)$ are \emph{safe}, meaning that the end state $(1,1)$ can be reached from $(x,y)$. In this example, we see that the whole space but the square $[U_n^A, U_m^A] \times [U_m^B, U_n^B]$ is reachable and the whole space but the square $[L_m^A, L_n^A] \times [L_n^B, L_m^B]$ is safe.

\subsection{Future work: semantics in directed spaces}

In future work, we hope to construct an interpretation of $\homt$-types into a suitable category of directed spaces (see e.g., \cite{Grandis, Krishnan2015}). We roughly sketch that interpretation here.

For any directed space $X$, we let $X^\core$ and $X^\op$ have the same underlying space as $X$, we let $X^\core$ have only trivial directed paths, and we reverse the directed paths of $X$ to get the directed paths of $X^\op$. Then we let the underlying space of $\homt_X(x,y)$ be the space of directed paths from $x$ to $y$ in $X$, and we let a directed path of paths 
from $p$ to $q$ in $\homt_X(x,y)$ be a homotopy of paths from $p$ to $q$ which is directed in the homotopy coordinate. The difficult and interesting part of establishing such an interpretation will lie in constructing a universe in this category.

We end by demonstrating that in such an interpretation of $\homt$-types into a category of directed spaces, together with an interpretation of Martin-L\"of's $\Sigma$-types, we will be easily able to express the notions of reachable and safe states. Indeed, given a type $X$ and two terms $i,f$ meant to represent the initial and final state, we can define
\[ \mathtt{Reachable}(X):= \Sigma_{x \in X} \homt_X(i,x) \]
\[  \mathtt{Safe}(X):= \Sigma_{x \in X^\op} \homt_X(x,f). \]
Then, for our directed space $S$ shown above in Figure \ref{fig:swissflag}, we would have that $\mathtt{Reachable}(S)$ is the directed space of reachable states each with a witness that that state is reachable and $\mathtt{Safe}(X)$ would be the directed space of safe states each with a witness that that state is safe.

\bibliographystyle{alpha}
\bibliography{directedhott}

\end{document}